\documentclass[twocolumn,conference]{IEEEtran}
\usepackage{bbm,dsfont,mathrsfs,fixmath,amsthm}
\usepackage{amsmath,amssymb,eucal,graphicx}
\usepackage{stmaryrd,cite}
\usepackage{amsmath,amstext,amsfonts,amssymb}
\usepackage{epsfig}
\usepackage{exscale}
\usepackage{enumerate}
\usepackage{mathtools}
\usepackage{multirow}

\newtheorem{lemma}{Lemma}
\newtheorem{cor}{Corollary}
\newtheorem{theorem}{Theorem}

\newcommand{\cond}{\,\vert\,}

\newcommand{\defeq}{\triangleq}

\newfont{\bbb}{msbm10 scaled 500}

\newfont{\bb}{msbm10 scaled 1100}

\newcommand{\FF}{\mbox{\bb F}}

\newcommand{\Ic}{{\cal I}}
\newcommand{\Jc}{{\cal J}}
\newcommand{\Kc}{{\cal K}}

\newcommand{\Sc}{{\cal S}}

\newcommand{\Wc}{{\cal W}}

\newcommand{\eqdef}{\stackrel{\Delta}{=}}

\DeclareFontFamily{U}{cmfi}{}
\DeclareFontShape{U}{cmfi}{m}{n}{ <-> cmfi10 }{}
\DeclareSymbolFont{CMFI}{U}{cmfi}{m}{n}

\begin{document}

\title{Cache-Enabled Broadcast Packet Erasure Channels with State Feedback}

\author{\authorblockN{}
\authorblockA{Asma Ghorbel, Mari Kobayashi, and Sheng Yang \\
LSS, CentraleSup\'elec \\
Gif-sur-Yvette, France\\
 {\tt \{asma.ghorbel, mari.kobayashi, sheng.yang\}@centralesupelec.fr}
}
}

\maketitle

\begin{abstract}
 We consider a cache-enabled $K$-user broadcast erasure packet channel in
 which a server with a library of $N$ files wishes to deliver a requested
 file to each user who is equipped with a cache of a finite memory $M$. Assuming that the transmitter has state feedback and user caches can be filled during off-peak hours reliably by decentralized cache placement, we characterize the optimal rate region as a function of the memory size, the erasure probability.  
The proposed delivery scheme, based on the scheme proposed by Gatzianas et al., exploits the receiver side information established during the placement phase. Our results enable us to quantify the net benefits of decentralized coded caching in the presence of erasure. The role of state feedback is found useful especially when the erasure probability is large and/or the normalized memory size is small.
\end{abstract}

\section{Introduction}
The exponentially growing mobile data traffic is mainly due to video applications (e.g. content-based video 
streams). Such video traffic has interesting features characterized by its asynchronous and skew nature. Namely, the user demands are highly asynchronous (since they request when and where they wish) and a few very popular files are requested over and over. The skewness of the video
traffic together with the ever-growing cheap on-board storage memory suggests
that the quality of experience can be boosted by 
caching popular contents at (or close to) end users in wireless networks. A number of recent works have studied such concept under different models and assumptions (see \cite{golrezaei2011femtocaching,maddah2013fundamental,ji2013fundamental} and references therein). In most of these works, it is assumed that the caching is performed in two phases: {\it placement phase} to prefetch users' caches under their memory constraints 
(typically during off-peak hours) prior to the actual demands; {\it delivery phase} to transmit codewords 
such that each user, based on the received signal and the contents of its
cache, is able to decode the requested file. 
In this work, we focus on a coded caching model where a content-providing server is connected to many users, each equipped with a cache of finite memory \cite{maddah2013fundamental}. By carefully choosing the sub-files to be distributed
across users, coded caching exploits opportunistic multicasting such that a common signal is simultaneously
useful for all users even with distinct file requests. A number of extensions of \cite{maddah2013fundamental} have been developed including the case of decentralized placement phase \cite{maddah2013decentralized}, the case of non-uniform demands \cite{niesen2013coded,ji2015order, hachem2015effect}, the case of unequal file sizes \cite{zhangcoded}. Although the potential merit of coded caching has been highlighted in these works, many of them have ignored the inherent features of wireless channels. 

The main objective of this work is to quantify the benefit of coded caching
by relaxing the unrealistic assumption of a perfect shared link. To this end, we model the bottleneck link as a broadcast packet erasure channel (BPEC) to capture random failure or disconnection of any server-user link that a packet transmission may experience especially during high-traffic hours (deliverly phase). The placement phase is performed either in a decentralized  \cite{maddah2013decentralized} or centralized manner \cite{maddah2013fundamental} over the erasure-free shared link.  
We further assume that the broadcast packet erasure channel is memoryless and
independent and identically distributed~(i.i.d.) across users and that the server acquires the channel states causally via feedback sent by users. Under this setting, we study the achievable rate region of the cache-enabled BPEC as a function of the main system parameters. Our contributions are two-hold: 1) a comprehensive analysis of the algorithm proposed by Gatzianas et al. \cite{gatzianas2013multiuser}, hereafter called GGT algorithm, which enables to characterize the achievable rate of high-order packet transmission; 2) characterization of the achievable rate region of the BPEC with feedback under decentralized cache placement for the case of distinct file requests.  
We prove that a simple delivery scheme extending GGT algorithm to the case of receiver side information can achieve the rate region.

Finally, we remark that a few recent works
\cite{huang2015performance,timo2015joint} have considered coded caching by
relaxing the perfect shared link assumption during delivery phase along the
line of this work. On the one hand, the work \cite{huang2015performance}
studies the resource allocation problem by modeling the bottleneck link as
multi-carrier fading channels. On the other hand, the authors in
\cite{timo2015joint} characterize the information theoretic tradeoff between
the reliable communication rate and the cache sizes in the erasure broadcast
channel with asymmetric erasure probability across users. In both works, it
is found that the performance of coded caching is somehow limited by the
worst case user. Contrary to this rather pessimistic conclusion, we find that
state feedback is useful to improve the performance of coded caching
especially in the regime of a small memory size~(with respect to the number
of files) and with a large erasure probability. This is because the packets not received by the intended users but overheard by unintended users can create multicast opportunity for later transmission at the price of a delay.  

The structure of the paper is as follows. Section \ref{section:MainResult}
introduces first the system model and definitions, and then highlights the main
results of this work. Sections \ref{section:converse} and
\ref{section:achievability} prove the converse and the achievability of the rate region of the cache-enabled BPEC with feedback, respectively. Section \ref{section:Examples} provides some numerical results to show the performance of our proposed delivery scheme and finally section \ref{section:conclusions} concludes the paper.  Throughout the paper, we use the following notational conventions.  
The superscript notation $X^n$ represents a sequence $(X_1,\ldots,X_n)$  of variables.
$X_{\Ic}$ is used to denote the set of variables $\{X_i\}_{i\in\Ic}$. 
Logarithm is to the base
$2$. The entropy of $X$ is denoted by $H(X)$. We define $\Ic_k \defeq \{1,\dots, k\}$ for $k=1,\dots, K$ and let $[K]=\{1,\dots, K\}$.

\section{System Model and Main Results} \label{section:MainResult}
\subsection{System model and definitions}
\begin{figure}
\vspace{-10pt}
\begin{center}
\includegraphics[width=0.45\textwidth,clip=]{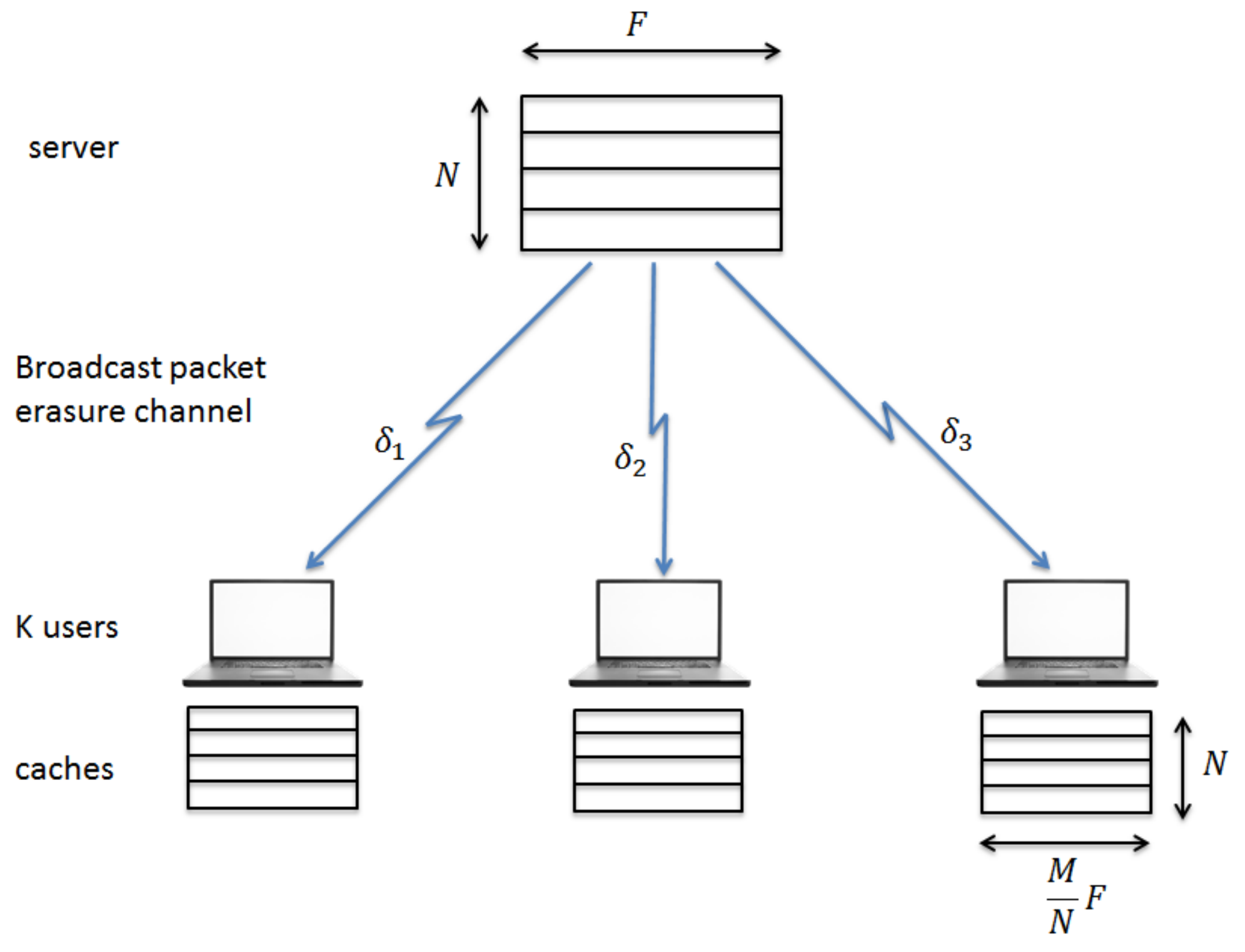}
\vspace{-2pt}
\caption{A cached enabled broadcast packet erasure channel for the case of $K=3$ and $F_i=F$ for all $i$}
\label{fig:model}
\end{center}
\vspace{-15pt}
\end{figure}

We consider a cache-enabled network depicted in Fig. \ref{fig:model} where a server is connected to $K$ users through a broadcast packet erasure channel (BPEC). 
The server has an access to $N$ files $W_1, \dots, W_N$ where file $i$, i.e.
$W_i$, consists of $F_i$ packets of $L$ bits each ($F_i L$ bits). Each user
$k$ has a cache memory $Z_k$ of $MF$ packets for $M\in [0,N]$, where
$F\defeq \frac{1}{N} \sum_{i=1}^N F_i$ is the average size of the files. 
Under such a setting, consider a discrete time communication system where a packet is sent in each slot over the $K$-user BPEC. 
The channel input $X_i\in  \FF_q$ belongs to the input alphabet of size $L=\log q$ bits\footnote{Throughout the paper, we assume that $L> \log_2 K$ so that the achievability results of \cite{gatzianas2013multiuser} hold.}. The channel is assumed to be memoryless and i.i.d. across users so that in a given slot we have 
\begin{align}\label{eq:EBC}
& \Pr( Y_{1}, Y_{2}, \dots, Y_{K}|X ) = \prod_{k=1}^K \Pr(Y_k | X) \\
& \Pr ( Y_k |X) =\begin{cases}
1-\delta,& Y_k=X,\\
\delta,& Y_k = E
\end{cases}
\end{align}
where $Y_k$ denote the channel output of receiver $k$ and $E$ denotes erasure. We let $S_i\in \Sc=2^{\{1,\dots, K\}}$ denote the
state of the channel in slot $i$ which indicates the users who received
correctly the packet. We assume that the transmitter obtains the state
feedback $S^{i-1}$ at the end of slot $i$ while all the receivers know $S^n$ at the end of the transmission. 

The caching is performed in two phases: placement phase and delivery phase.
In placement phase, the server fills the caches of all users $Z_1, \dots,
Z_K$ up to the memory constraint. As in most works in the literature, we
assume that the placement phase is done without error and neglect the cost,
since it takes place usually during off-peak traffic hours. Once each user
$k$ makes a request $d_k$, the server sends codewords so that each user can
decode its requested file as a function of its cache content and received
signals during delivery phase. We provide a more formal definition below.  A
$(M, F_{d_1}, \dots, F_{d_K}, n)$ caching scheme consists of the following
components.
\begin{itemize}
\item $N$ message files $W_1,\dots, W_N$ are independently and uniformly distributed over $\Wc_1 \times \dots \times \Wc_N$ with $\Wc_i = \FF_q^{F_i}$ for all $i$.
\item $K$ caching functions are given by $\phi_k: \FF_q^{\sum_{i=1}^N F_i} \to\FF_q^{FM}$ map the files $W_1, \dots, W_N$ into the cache contents 
\begin{align}
Z_k \defeq \phi_k(W_1,\dots, W_N)
\end{align}
for each user $k$. 
\item A sequence of encoding functions which transmit at slot $i$ a symbol
  $X_i = f_i (W_{d_1}, \dots, W_{d_K}, S^{i-1})\in \FF_q$, based on the
  requested file set and the channel feedback up to slot $i-1$ for $i=1,
  \dots, n$, where $W_{d_k}$, $d_k \in \{\emptyset, 1, \dots, N\}$, denotes the message file requested by user $k$. 
\item A decoding function of user $k$ is given by the mapping $g_k: \FF_q^n
  \times\FF_q^{FM} \times \Sc^n \to \FF_q^{F_{d_k}}$ so that
the decoded file is $\hat{W}_{d_k}= g_k(Y_k^n, Z_k, S^n)$ 
as a function of the received signals $Y_k^n$, the cache content $Z_k$, as well as the state information $S^n$. 
\end{itemize}
A rate tuple $(R_1,\dots, R_K)$ is said to be achievable if, for every
$\epsilon>0$, there exists a $(M, F_{d_1},\dots, F_{d_K}, n)$ caching strategy that satisfies 
\begin{itemize}
\item reliability condition 
\[  \max_{(d_1,\dots, d_K)\in \{1,\dots, N\}^K} \max_k  \Pr( g_k(Y_k^n, Z_k, S^n) \neq W_{d_k} ) < \epsilon \]
\item rate condition 
\begin{align}
 R_k  <  \frac{F_{d_k}}{n}. 
\end{align} 
\end{itemize}
Throughout the paper, we express the entropy and the rate in terms of packets in oder to avoid the constant factor $L=\log_2 q$. 
In this work, we focus on the case of equal file size $F_i=F$ for simplicity. 

\subsection{Decentralized cache placement}
We mainly focus on decentralized cache placement proposed in \cite{maddah2013decentralized} and adapt it to the packet-based broadcast channel (with no error). Under the memory constraint of $MF$ packets, each user $k$ independently caches a subset of $\frac{MF}{N}$ packets of file $i$, chosen uniformly at random for $i=1,\dots, N$. By letting $W_{i|\Kc}$ denote the sub-file of $W_i$ stored in the cache memories (known) of the users in $\Kc$, the cache memory $Z_k$ of user $k$ after decentralized placement is given by
\begin{align} \label{eq:Zk}
Z_k =\{ W_{i \cond \Kc} \;\; \forall k \subseteq \Kc \subseteq [K], \;\;  \forall i =1,\dots, N. \} 
\end{align}
To illustrate the placement strategy, consider an example of $K=3$ users and a file $A$ of $F$ packets. 
After the placement phase, a given file A will be partitioned into 8 subfiles:
\begin{align}\label{eq:DecentralizedPlacement}
A=\{A_0, A_1, A_2, A_3, A_{12}, A_{13}, A_{23}, A_{123} \} 
\end{align}
where, for $\Kc \subset \{1,2,3\}$,  $A_{\Kc}$ denotes the packets of file $A$ stored exclusively in the cache memories of users in $\Kc$. By a law of large numbers as $F\rightarrow \infty$, the size of $|A_{\Kc}|$ measured in packets is given by 
\begin{align}
  \frac{|A_{\Kc}|}{F}=p^{|\Kc|}(1-p)^{3-|\Kc|}. \label{eq:LLN}
\end{align}%

\subsection{Main results}
In order to characterize the rate region of a cached-enabled BPEC with state feedback, we focus on the case of most interest with $N\geq K$ and assume further that users' demands are all distinct. 
\begin{theorem} \label{theorem:region}
The optimal rate region of the cached-enabled BPEC with state feedback under decentralized cache placement is given by
\begin{align}\label{eq:wsr}
\sum_{k=1}^K \frac{\left(1-\frac{M}{N}\right)^k}{1-\delta^k} R_{\pi_k} \leq 1
\end{align}
for any permutation $\pi$ of $\{1,\dots, K\}$. 
\end{theorem}
The proof of Theorem \ref{theorem:region} is provided in upcoming sections. 
This region yields the following symmetrical rate $R_k=R_{\rm sym}(K) $ for
all $k$ with
\begin{align}\label{eq:symmetricR}
 R_{\rm sym} (K)= \frac{1}{\sum_{k=1}^K
 \frac{\left(1-\frac{M}{N}\right)^k}{1-\delta^k}}.
\end{align} 
The following corollary holds. 
\begin{cor}\label{cor:rate}
The minimum number of transmissions to deliver a distinct file to each user in the cached-enabled BPEC under decentralized cache placement is given by
\begin{align}
T_{\rm tot}=   \Theta(F) \sum_{k=1}^K \frac{\left(1-\frac{M}{N}\right)^k}{1-\delta^k} 
\end{align}
as $F\rightarrow \infty$. 
\end{cor}
The following remarks are in order. The results cover some special cases of interest. For the case without cache memory $M=0$, 
the region in Theorem \ref{theorem:region} simply boils down to the BPEC with state feedback \cite{gatzianas2013multiuser}. For the case
of no erasure, the number of transmission in Corollary \ref{cor:rate} scaled
by $F$ is precisely the rate-memory tradeoff under decentralized cache
placement for $N\geq K$ \cite{maddah2013decentralized}. In fact, after some
simple algebra, the number of transmissions normalized by the file size can be rewritten as 
\begin{align}
\frac{T_{\rm tot}}{F}=  \frac{N}{M} \left(1-\frac{M}{N}\right) \left\{1-\left(1-\frac{M}{N}\right) ^K\right\}.
\end{align}
This coincides with the ``rate'' measured as the number of files to be sent
over the shared perfect link as defined by Maddah-Ali and Niesen
\cite{maddah2013decentralized}.  In the following sections we provide the proof
of the converse and achievability.

\section{Optimality of delivery phase} \label{section:converse}
In this section, we provide the converse part of Theorem \ref{theorem:region} under the assumption that decentralized cache placement is performed. 
We let $p=\frac{M}{N}$ denote the probability of storing a file in a given user's cache memory. First we provide two useful lemmas. 
\begin{lemma}\cite[Lemma 5]{ShengISIT2015}\label{lemma:erasure-ineq}
  For the broadcast erasure channel with independent erasure events~(with probability $\{\delta_k\}$) for
  different users, if $U$ is such that $X_i\leftrightarrow U  Y_{\Ic}^{i-1}  S^{i-1} \leftrightarrow
  (S_{i+1},\ldots,S_n)$, $\forall\,\Ic$,
\begin{multline}
  \frac{1}{1-\prod_{i\in \Ic} \delta_i} H(Y^n_{\mathcal{I}} \cond U, S^n) \le \frac{1}{1-\prod_{i\in \Jc} \delta_i }
  H(Y^n_{\mathcal{J}} \cond U, S^n), 
  \label{eq:essential-erasure}
\end{multline}%
for any sets $\Ic, \Jc$ such that $\Jc \subseteq \Ic \subseteq \left\{ 1,\ldots,K \right\}$.  
\end{lemma}
\begin{proof}
Appendix \ref{appendix:erasure-ineq}.
\end{proof}
\begin{lemma} \label{lemma:decentralized}
Under decentralized cache placement \cite{maddah2013decentralized}, the
following equality holds for any $i$ and $\Kc\subseteq[K]$  
\[ H(W_i \cond \{Z_{k}\}_{k\in \Kc} ) = \left(1-p\right) ^{|\Kc|} H(W_i). \]
\end{lemma}
\begin{proof}
\begin{align}
  \MoveEqLeft{H(W_i \cond \{Z_{k}\}_{k\in \Kc} )} \nonumber \\ 
  & = H(W_i \cond  \{W_{l |\Jc} \}_{ \Jc: \Jc\cap \Kc\ne \emptyset,\, l=1,\dots, N})\\
  & = H(W_i \cond  \{ W_{i |\Jc} \}_{ \Jc: \Jc\cap \Kc\ne \emptyset} )\\
  & = H( \{ W_{i |\Jc} \}_{ \Jc: \Jc\cap \Kc = \emptyset} ) \\
  & = \sum_{ \Jc: \Jc\cap \Kc = \emptyset} H(  W_{i |\Jc} ) \label{eq:tmp892}\\
  & = \sum_{ \Jc: \Jc\cap \Kc = \emptyset} p^{|\Jc|}(1-p)^{K-|\Jc|}H( W_{i} ) \label{eq:tmp893} \\
  & = H( W_{i} ) \ \sum_{l=0}^{K-|\Kc|} {K-|\Kc| \choose l}  p^{l}(1-p)^{K-l} \\
  & = (1-p)^{|\Kc|} H(W_{i}) 
\end{align}
where the first equality follows from \eqref{eq:Zk}, 
the second equality follows due to the independence between message files,
the third equality follows by identifying the unknown parts of $W_i$ given
the cache memories of $\Kc$ and using the independence of all sub-files;
\eqref{eq:tmp892} is again from the independence of the sub-files;
\eqref{eq:tmp893} is from the law of large number similarly as in
\eqref{eq:LLN}; finally the last equality follows by applying the binomial
theorem. 
\end{proof}

We apply genie aided bounds to create a degraded erasure broadcast channel by
providing the messages, the channel outputs, as well as the receiver side
information (contents of cache memories) to enhanced receivers. We focus on
the case without permutation and the demand $(d_1,\dots, d_K)=(1,\dots, K$)
due to the symmetry. We have for user $k$, $k=1,\ldots,K$,
\begin{align}
n (1-p)^k R_k &= (1-p)^k H(W_k) \\
 &= H(W_k |Z^k S^n) \\
 &\leq  I(W_k;Y_{\Ic_k}^n \cond Z^k S^n) + n \epsilon'_{n,k} \label{eq:tmp722}\\
 &\leq  I(W_k;Y_{\Ic_k}^n, W^{k-1} \cond Z^k S^n) + n \epsilon'_{n,k} \\
 &=  I(W_k;Y_{\Ic_k}^n \cond  W^{k-1} Z^k S^n) + n \epsilon'_{n,k} 
\end{align}
where the second equality is by applying Lemma \ref{lemma:decentralized} and
noting that $S^n$ is independent of others, \eqref{eq:tmp722} is from the
Fano's inequality; the last equality is from $I(W_k; W^{k-1} \cond
Z^k S^n) = 0$. Putting all the rate constraints together, and letting
$\epsilon_{n,k} \defeq \epsilon'_{n,k}/(1-p)^k$,
\begin{align}
  n (1-p)(R_1 - \epsilon_{n,1}) &\leq   H(Y^n_1 \cond Z_1 S^n) - H(Y^n_1\cond
  W_1  Z_1 S^n) \nonumber \\
&\ \, \vdots \nonumber \\
n (1-p)^K (R_K - \epsilon_{n,K}) &\le H(Y^n_{\mathcal{I}_K} \cond W^{K-1} Z^{K}  S^n) \nonumber \\
&\qquad - H(Y^n_{\mathcal{I}_K}\cond W^{K} Z^{K} S^n)
\end{align}%
We now sum up the above inequalities with different weights, and applying
Lemma~\ref{lemma:erasure-ineq} for $K-1$ times, namely, for $k=1,\ldots,K-1$,
\begin{align}
  \frac{H(Y^n_{\mathcal{I}_{k+1}} \cond W^{k} Z^{k+1} S^n) }{1-\delta^{k+1}}  &\le  \frac{H(Y^n_{\mathcal{I}_{k+1}} \cond W^{k} Z^{k} S^n)}{1-\delta^{k+1}}  \\ 
  &\le \frac{H(Y^n_{\mathcal{I}_{k}} \cond W^{k} Z^{k} S^n)}{1-\delta^k} 
\end{align}%
where the first inequality follows because removing conditioning increases the entropy. Finally, we have 
\begin{align}
\MoveEqLeft{\sum_{k=1}^{K} \frac{(1-p)^k}{1-\delta^k} (R_k - \epsilon_n)}
\nonumber \\
&\le \frac{H(Y^n_1  \cond Z_1 S^n)}{n(1-\delta)} - \frac{H(Y^n_{\mathcal{I}_K} \cond W^{K} Z^K S^n)}{n(1-\delta^K)}
\\
&\leq \frac{H(Y^n_1) }{n(1-\delta)}\le1
\end{align}%
which establishes the converse proof.

\section{Achievability} \label{section:achievability}
Exploiting the polyhedron structure, the vertices of the rate region \eqref{eq:wsr} can be proven to be 
\begin{align}
R_k = \begin{cases}
R_{\rm sym}( |\Kc|), k \in \Kc \\
0, k\notin \Kc 
\end{cases}
\end{align}
for $\Kc \subseteq [K]$. The proof follows the same footsteps as \cite[Section V]{maddah2010degrees} and shall be omitted. This means that when only $|\Kc|$ users are active in the system, each active user achieves the same symmetrical rate as the reduced system of dimension $|\Kc|$. Then, it suffices to prove the achievability of the symmetrical rate for a given dimension $K$.  To this end, we first revisit the algorithm proposed by Gatzianas et al. \cite{gatzianas2013multiuser}, hereafter called GGT algorithm, and characterize the high-order transmission rates. Then, we extend GGT algorithm to the context of the cached-enabled broadcast erasure packet channel.
\subsection{GGT algorithm revisited}
The algorithm proposed by Gatzianas et al. \cite{gatzianas2013multiuser} consists of $K$ phases. 
In each phase $k$, the transmitter sends order-$k$ packets simultaneously useful to a subset of $k$ users. A phase $k$ 
is further partitioned into $K \choose k$ subphases in each of which the transmitter sends packets intended to a unique subset of $k$ users. 
We wish to provide an informal but intuitive description of GGT algorithm along the line of \cite{maddah2010degrees} by assuming the number of private packets $N_0$ per user is arbitrarily large so that the length of each phase becomes deterministic. First we introduce the basic notions together with key parameters. 
To simplify the description of the algorithm, we let $i, j$ denote the cardinality of $\Ic, \Jc$. \begin{itemize}
\item $t_j$ denotes the duration of a given subphase intended to $j$ users in slots. The duration of phase $j$ is given by
$T_j={K \choose j} t_j$.
\item A packet of order-$i$ becomes order-$j$ {\it for a given user} for $i<j\leq K$ if erased by this user and all users in $[K]\setminus \Jc$ but received by $\Jc \setminus \Ic$. The probability of this event is denoted by $\alpha_{i \rightarrow j} = \delta^{K-j+1} (1-\delta)^{j-i}$.
We let 
\begin{align}\label{eq:Nij}
N_{i\rightarrow j}=t_i \alpha_{i\rightarrow j}
\end{align}
 denote the number of such packets. 
\item An order-$j$ packet is consumed {\it for a given user} if this user or at least one user in $[K]\setminus \Jc$ receives it. The probability
of this event is denoted by $\beta_j=1-\delta^{K-j+1}$. 
\end{itemize}
Due to the symmetry across users, we can focus on any arbitrary user to
define the parameters $N_{i\rightarrow j}$ and $\beta_j$. Under this setting, the length of order-$j$ subphase is given recursively by
\begin{align}\label{eq:RecursiveLength}
t_j = \frac{1}{\beta_j} \sum_{i=1}^{j-1} {j-1 \choose i-1}N_{i\rightarrow j}.
\end{align}
Here is a brief summary of the algorithm:
\begin{enumerate}
\item Phase $1$ (order-1 transmission): send $N_0$ private uncoded packets to each user. This generates $N_{1\rightarrow j}$ order-$j$ packets to be sent during phase $j=2,\dots, K$. 
\item Phase $j$ (order-$j$ transmission) for $j=2,\dots, K$: in each subphase intended to a subset $\Jc$ of users, send random linear combinations\footnote{The exact generation method to generate the linear combination is explained in \cite{gatzianas2013multiuser} and shall not be repeated here.} of $\sum_{i=1}^{j-1} {j \choose i} N_{i\rightarrow j}$ packets for all $\Ic \subseteq \Jc$. This subphase generates $N_{j\rightarrow j'}$ order-$j'$ packets to be sent in phase $j'>j$. Proceed sequentially for all subsets $\Jc$ of cardinality $j$. 
\end{enumerate}

\begin{figure}
\vspace{-10pt}
\begin{center}
\includegraphics[width=0.45\textwidth,clip=]{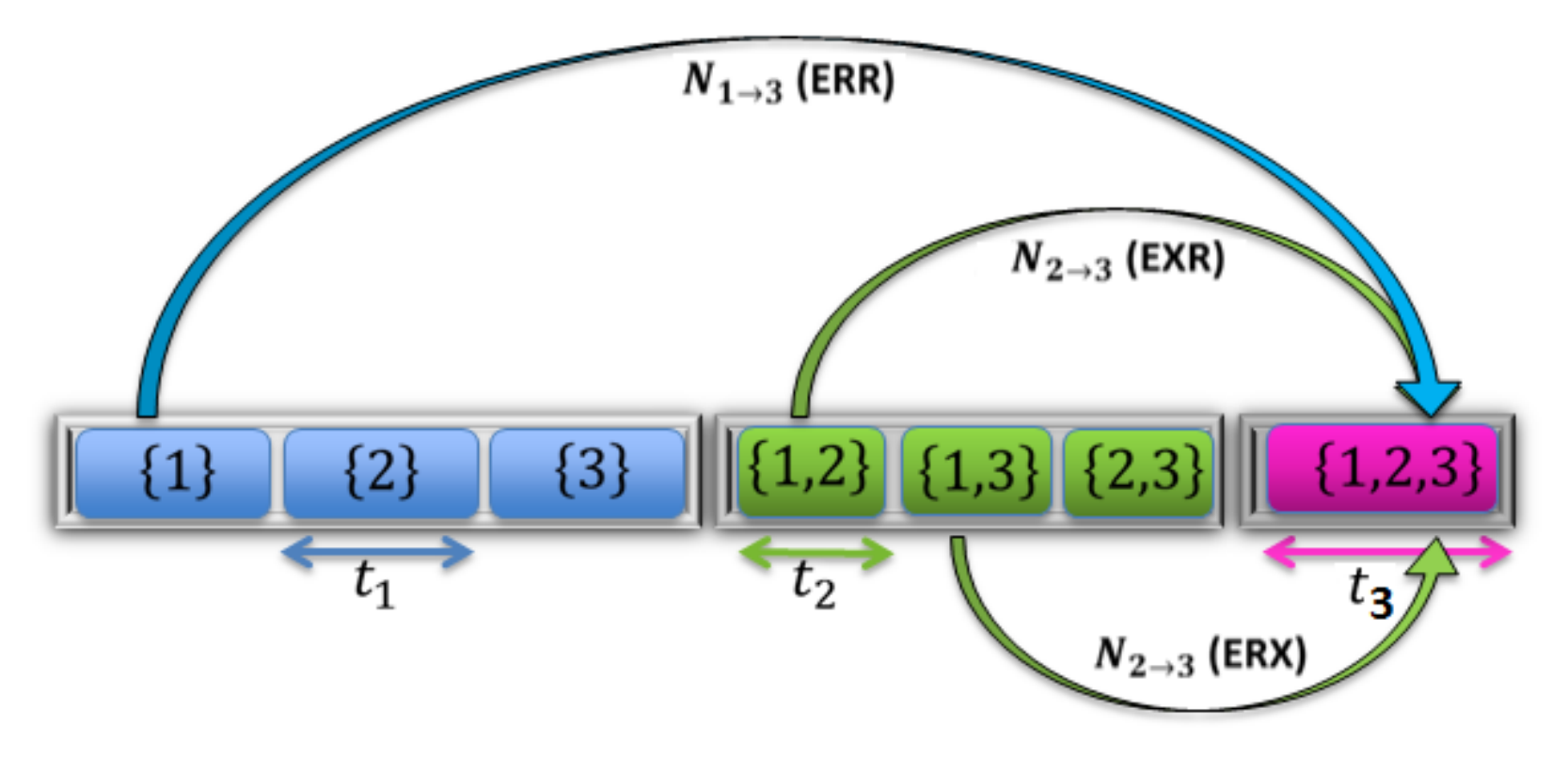}
\vspace{-2pt}
\caption{Phase organization for $K=3$ and packet evolution viewed by user 1.}
\label{fig:phase3}
\end{center}
\vspace{-10pt}
\end{figure}

Fig. \ref{fig:phase3} illustrates the phase (subphase) organization for $K=3$ and the packet evolution viewed by user 1. The order-$3$ packet is created both from phases 1 and 2. 
More precisely, the order-$1$ packets for user 1 becomes order-$3$ (via linear combination) if erased by user 1 and received by others ($ERR$). The number of such packets is $N_{1\rightarrow 3}$. 
Order-$2$ packets 
intended to $\{1,2\}$ becomes order-$3$ if erased by user 1 but received by user 3 ($EXR$) while packets intended to $\{1,3\}$ become order-$3$ if 
erased by user 1 and received by user 2 (event $EXR$). The total number of order-$3$ packets created from phase 2 is $2N_{2\rightarrow 3}$. 

\begin{lemma} \label{lemma:order}
In the $K$-user erasure broadcast channel with feedback, the sum rate of order-$i$ packets, denoted by $R^i(K)$ is upper bounded by
\begin{align}\label{eq:Rorder}
R^i (K) &\leq   \frac{{K \choose i}}{\sum_{k=1}^{K-i+1} \frac{ {K-k \choose i-1}}{1-\delta^k}}.
\end{align}
Algorithm GGT achieves the RHS of \eqref{eq:Rorder} with equality. 
\end{lemma}
\begin{proof}
Appendix \ref{appendix:orderRate}.
\end{proof}
As a corollary of Lemma \ref{lemma:order}, the symmetrical rate \eqref{eq:symmetricR} can be rewritten in a convenient form.
\begin{cor}\label{cor:Sheng}
The order-$1$ rate, $R^1(K) = K R_{\rm sym}$ of the $K$-user broadcast erasure channel with feedback can be expressed as a 
function of $R^2, \dots, R^{K}$ as follows. 
\begin{align}\label{eq:Sheng}
R^1(K)  = \frac{KN_0 }{ \frac{K N_0}{\beta_1} + \sum_{j=2}^K \frac{{K\choose j} N_{1\rightarrow j}}{R^j(K)}}
\end{align}
where $\frac{K N_0}{\beta_1} $ is the duration of phase 1, ${K \choose j} N_{1\rightarrow j}$ corresponds to the total number of order-$j$ packets generated in phase 1.
\end{cor}
\begin{proof}
Appendix \ref{appendix:Sheng}.
\end{proof}

\subsection{Proposed delivery scheme}
Now we are ready to describe the delivery phase by extending GGT algorithm to the context of the cached-enabled network. 
For simplicity, we first provide an example of $K=3$ users and three files $A, B, C$ of size $F$ packets each. 
After decentralized placement phase, each file is partitioned into $8$ subfiles as seen in \eqref{eq:DecentralizedPlacement}. 
Obviously, the subfile $A_{\Jc}$ for $1\in \Jc$, i.e. $A_1, A_{12}, A_{13}, A_{123}$ are received by the destination and shall not be transmitted in delivery phase. The same holds for $B_{\Jc}$ for $2\in \Jc$ as well as $C_{\Jc}$ for $3\in \Jc$. 

In the presence of users' caches, the packets to be sent in phases $j$ are composed by order-$j$ packets created by the algorithm and placement phase. By treating placement phase as phase 0, we let $N_{0\rightarrow j}$ denote the number of order-$j$ packets generated in placement phase for a subset of $j$ users for $j\geq 2$. For $j=1$, we use a short-hand notation $N_{0\rightarrow 1} = N_0$. By focusing on user 1, we have 
\begin{align}\label{eq:N0j}
N_0 &= |A_0|=F (1-p)^3, \;  N_{0\rightarrow 3} = |A_{23}| =F p^2(1-p) \nonumber \\
 N_{0\rightarrow 2} &= |A_2| = |A_3|= F p (1-p)^2.  
\end{align}

\paragraph{Phase 1} The transmitter sends $A_0, B_0, C_0$ in TDMA.
Each packet is repeated until at least one user receives it. 
After a subphase intended to user 1,  the subfile $A_0$ is partitioned into: $ A_0 =\{A_{01}, A_{02}, A_{03}, A_{012}, A_{013}, A_{023}, A_{0123} \}$
where $A_{0\Jc}$ denotes with some abuse of notation the part of $A_0$ received by receivers in $\Jc$ for $\Jc\subseteq\{1,2,3\}$.
In other words, a subphase creates order-$j$ packets whose number is given by  
\begin{align}\label{eq:N1j}
N_{1\rightarrow j+1} = |A_{0\Jc}| = \frac{N_0}{1-\delta^3} (1-\delta)^{j} \delta^{3-j}
\end{align}
for any $\Jc$ of cardinality $j=1,2$.  
\paragraph{Phase 2} The transmitter sends packets in three subphases for users $\{1,2\}$, $\{1,3\}$, $\{2,3\}$, where each subphase is of size
\[t_2 = \frac{N_{1\rightarrow 2} + N_{0\rightarrow 2}}{1-\delta^2}\]
\begin{itemize}
\item subphase $\{1,2\}$: linear combinations $F$ between $A_2, A_{02}$ and $B_1, B_{01}$
\item subphase $\{1,3\}$: linear combinations $G$ between $A_3, A_{03}$ and $C_1, C_{01}$
\item subphase $\{2,3\}$: inear combinations $H$ between $B_3, B_{03}$ and $C_2, C_{02}$
\end{itemize}
Phase 2 creates order-$3$ packets given as linear combinations of $F_{13}, F_{23}$, $G_{12}, G_{23}$, $H_{12}, H_{13}$.
The size of any of these packets is given by $N_{2\rightarrow 3} = |F_{13}| = t_2 \delta (1-\delta)^2$. 

\paragraph{Phase 3} Send order-$3$ packets obtained by linear combinations of $A_{23}, B_{13}, C_{12}$ created in placement phase, 
$A_{023}, B_{013}, C_{012}$ created in phase 1, and $F_{13}, F_{23}$, $G_{12}, G_{23}$, $H_{12}, H_{13}$ created in phase 2. 
The length of phase 3 is given by
\begin{align}
T_3 = t_3 =  \frac{  N_{0\rightarrow 3} +N_{1\rightarrow 3} +2 N_{2\rightarrow 3} }{1-\delta}. 
\end{align}

For the three-user example, it is possible to compute the length of each phase recursively to find the symmetrical rate. In fact, it suffices to consider additionally the packets generated in phase 0 by adding $N_{0\rightarrow j}$ in \eqref{eq:RecursiveLength}.
However, this straightforward approach is no longer feasible for a large $K$. Therefore, we apply Corollary \ref{cor:Sheng} to find the symmetrical rate.
From Lemma \ref{lemma:order}, we have $R^2(3)=\frac{1-\delta^2}{1+2(1+\delta)}$ and $R^3(3) = 1-\delta$. Plugging \eqref{eq:N0j}, \eqref{eq:N1j}, we
readily obtain 
\begin{align}
R^1_{\rm cache} (3) 
&=3 \left(\frac{1-p}{1-\delta} +\frac{(1-p)^2}{1-\delta^2}
+\frac{(1-p)^3}{1-\delta^3} \right)^{-1}.
\end{align}
which coincides with $3R_{\rm sym}(3)$.

A generalization to the $K$-user case is rather trivial since the placement phase only yields the high-order packets to be sent together with those generated by the algorithm. 
The achievable symmetrical rate is obtained by modifying \eqref{eq:Sheng} by
including the packets generated from placement phase as 
\begin{align}
R^1_{\rm cache}(K)  = \frac{K N_0}{\frac{KN_0}{\beta_1} + \sum_{j=2}^K
\frac{{K\choose j} (N_{0\rightarrow j} + N_{1\rightarrow j})}{R^j(K)} }.
\end{align} 
By repeating the same steps as the proof of Corollary \ref{cor:Sheng}, it readily follows that the above 
expression coincides with $KR_{\rm sym}(K)$.

\section{Numerical Examples}\label{section:Examples}
\begin{figure*}
  \begin{minipage}{\columnwidth}
\begin{center}
\includegraphics[width=0.95\textwidth,clip=]{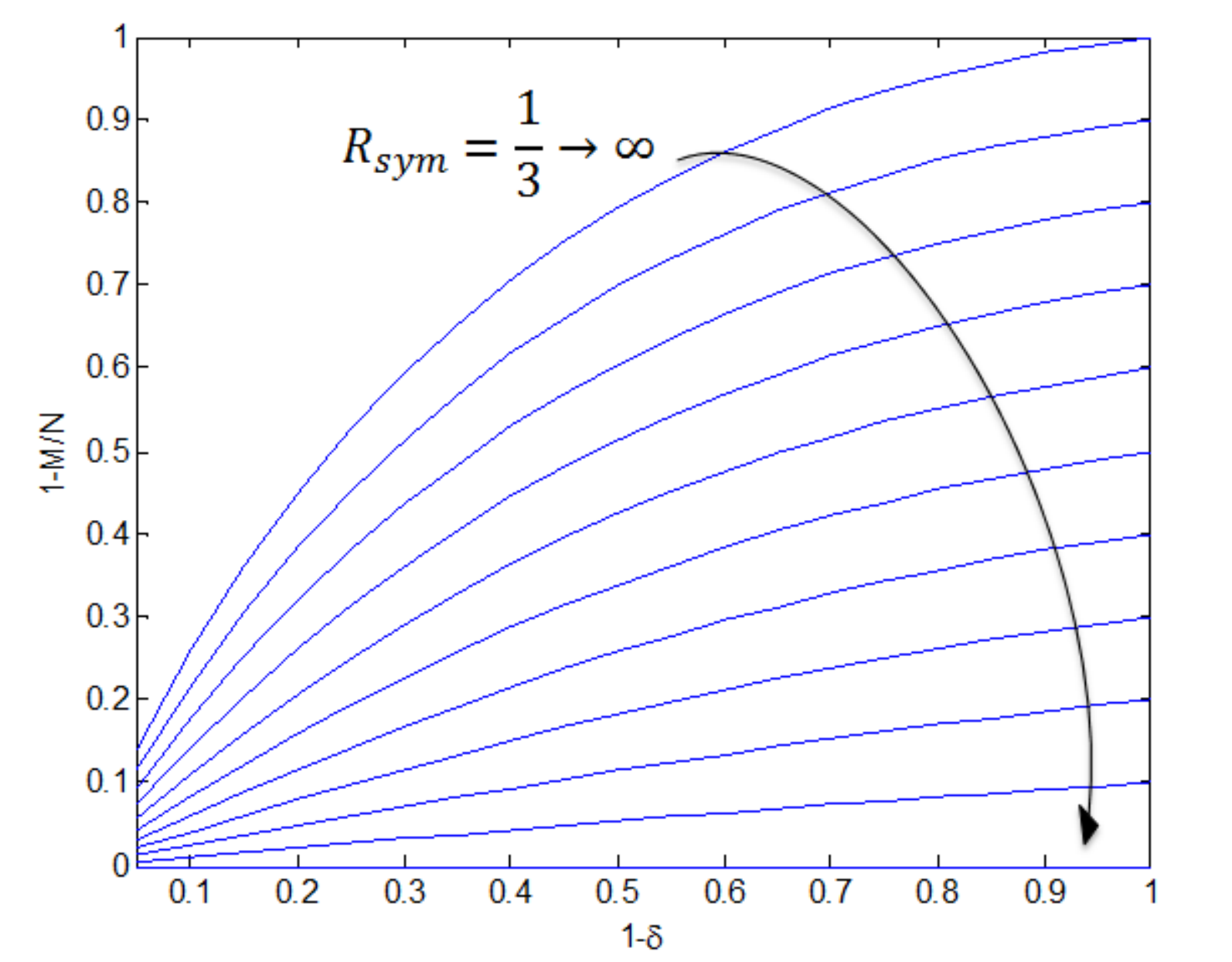}
\caption{The tradeoff between the memory and the erasure for $K=3$.}
\label{fig:memory-erasure}
\end{center}
\end{minipage}
  \begin{minipage}{\columnwidth}
\begin{center}
\includegraphics[width=0.9\textwidth,clip=]{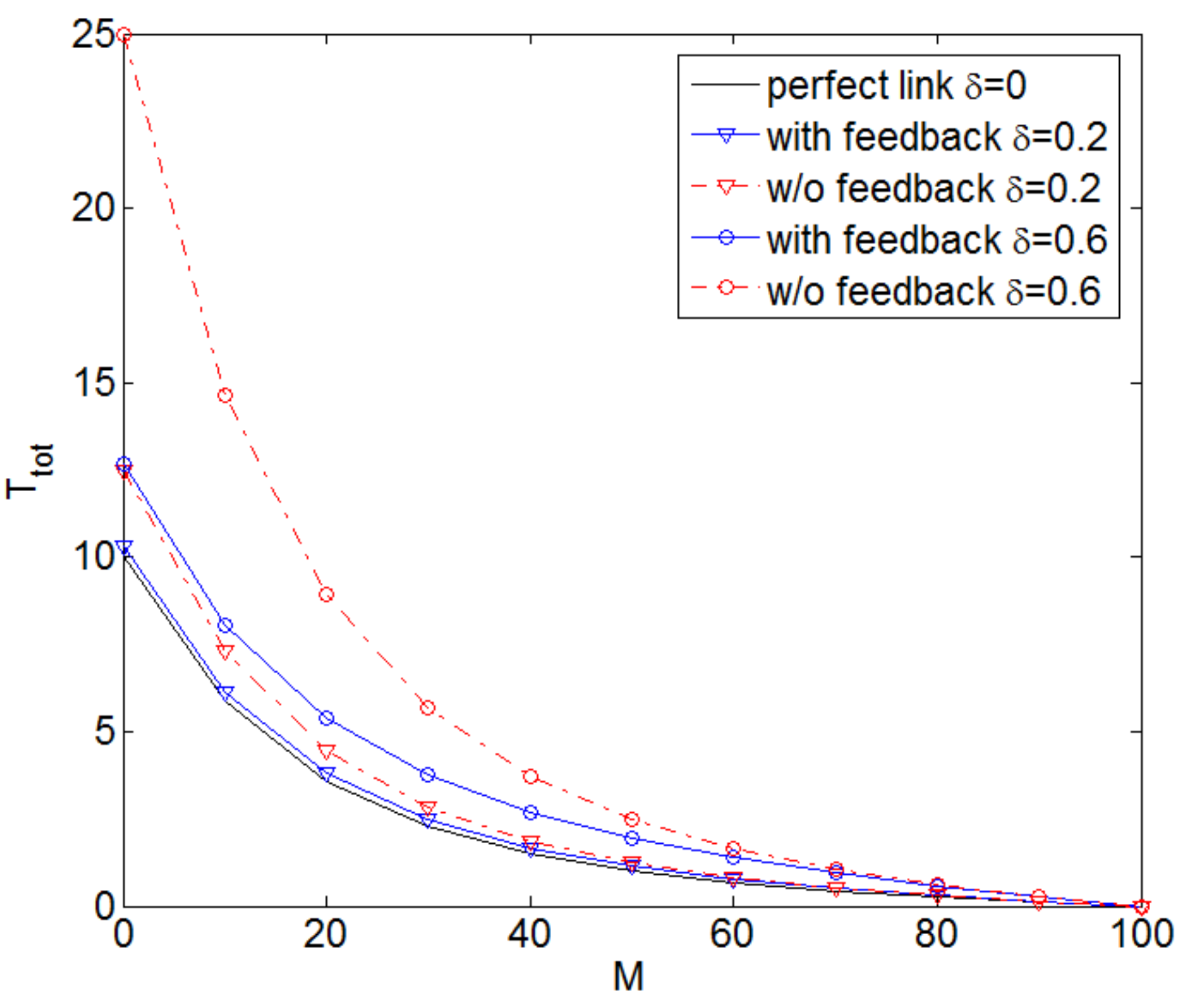}
\caption{The number of transmission $T_{\rm tot}$ as a function of memory size $M$ for $N=100, K=10$. }
\label{fig:TtotvsM}
\end{center}
\end{minipage}
\end{figure*}

In this section we provide some numerical examples to show the performance of our proposed delivery scheme. 
Fig.~\ref{fig:memory-erasure} illustrates the tradeoff between the erasure probability and the memory size for the case of $K=3$.
Each curve corresponds to a different symmetrical rate $R_{\rm sym}(3)$. The arrow shows the increasing $R_{\rm sym}$ from $1/3$, corresponding to the broadcast channel without memory $M=0$ and erasure $\delta=0$, to infinity. The cache memory increases the 
rate performance even in the presence of erasure and the benefit of memory
cache is significant for smaller erasure probabilities as expected 
from the analytical expression. Fig. \ref{fig:TtotvsM} compares the number of transmission $T_{\rm tot}$, normalized by the file size $F$, achieved by our delivery scheme with feedback 
and the scheme without feedback. We consider the system with $N=100, K=10$
and the erasure probabilities of $\delta=0$ (perfect link), $0.2$, and $0.6$.  We observe that state feedback can be useful especially when the memory size is small and the erasure probability is large. 
In fact, it readily follows that  the rate region of the cached-enabled broadcast channel with no feedback is given by
\begin{align}\label{eq:NoFB-region}
\sum_{k=1}^K \frac{\left(1-\frac{M}{N}\right)^k}{1-\delta} R_{\pi_k} \leq 1 
\end{align}
yielding
\begin{align}
T_{\rm tot-noFB} = F \frac{\sum_{k=1}^K \left(1-\frac{M}{N}\right)^k}{1-\delta}.
\end{align}
This corresponds to the total number of transmission over the perfect link expanded by a factor $\frac{1}{1-\delta}>1$ because any packet must be received by all users whatever the order of the packet is. 
Recalling that the feedback is useless for multicasting, the merit of feedback becomes significant if packets of lower order dominate the order-$K$ packets.
The case of small $p=\frac{M}{N}$ and large erasure probability corresponds
to such a situation.

\section{Conclusions}\label{section:conclusions}
In this work, we studied decentralized coded caching in the broadcast erasure packet channels with state feedback.
Our main contribution is the characterization of the achievable rate region of the channel at hand for the worst case demand such that
users' requests are all different in the regime of a large number of files $N\geq K$. 
Contrary to the pessimistic conclusion made by recent work \cite{huang2015performance,timo2015joint} , it is found that the performance of coded caching is no longer limited by the worst user in the presence of state feedback. In fact, state feedback is useful to improve
the rate performance especially when the erasure probability is large and/or the normalized memory size is small. 

While we restricted ourselves to some regime of interest and to decentralized cache placement for the sake of simplicity, our work can be easily extended to other regimes and the case of centralized coded caching. For example, let us consider the case where subsets of users request a common file and the transmitter must convey a mixture of different-order messages (typically in the regime of $N< K$), our proposed delivery scheme can be extended to such situation along the line of \cite{piantanida2013analog}. It should be noticed that the converse proof on the different-order message rate in Lemma \ref{lemma:order} is already general enough to cover all possible file requests. Interestingly, the proposed delivery algorithm can apply directly to centralized coded caching by starting phase-$(b+1)$ transmission where $b=\frac{KM}{N}$ is the ratio of the aggregate memory to the library size. Other interesting yet non-trivial generalization includes the case of non-uniform popularity distribution as well as the case of online coded caching.

\appendix
\section*{Elements of Proofs}
\subsection{Proof of Lemma~\ref{lemma:erasure-ineq}} \label{appendix:erasure-ineq}
  We have, for $\Jc\subseteq\Ic$,
  \begin{align}
    \MoveEqLeft[0]{H(Y^n_{\mathcal{I}} \cond U, S^n)}\\ 
    &= \sum_{i=1}^n H(Y_{\mathcal{I}, i} \cond
    Y^{i-1}_{\mathcal{I}}, U, S^n) \\
    &= \sum_{i=1}^n H(Y_{\mathcal{I},i} \cond Y^{i-1}_{\mathcal{I}}, U, S^{i-1}, S_i) \\
    &= \sum_{i=1}^n \mathrm{Pr}\{S_i\cap\mathcal{I}\ne\emptyset\} \, H(X_i \cond Y^{i-1}_{\mathcal{I}}, U, S^{i-1}, S_i\cap\mathcal{I}\ne\emptyset) \\
    &= \sum_{i=1}^n \bigl(1-\prod_{i\in\Ic}\delta_i\bigr) H(X_i \cond Y^{i-1}_{\mathcal{I}}, U, S^{i-1}) \\
    &\le \bigl(1-\prod_{i\in\Ic}\delta_i\bigr) \sum_{i=1}^n  H(X_i \cond Y^{i-1}_{\mathcal{J}}, U, S^{i-1})
    \label{eq:tmp821}
  \end{align}%
  where the first equality is from the chain rule; the second equality is due to the current input does not
  depend on future states conditional on the past outputs/states and $U$; the third one holds since $Y_{\mathcal{I},i}$ is deterministic and has entropy $0$
  when all outputs in $\mathcal{I}$ are erased~($S_i\cap\mathcal{I}=\emptyset$); the fourth equality is from
  the independence between $X_i$ and $S_i$; and we get the last inequality by removing the terms
  $Y^{i-1}_{\Ic\setminus\Jc}$ in the condition of
  the entropy. Following the same steps, we have 
  \begin{align}
    \MoveEqLeft[0]{H(Y^n_{\mathcal{J}} \cond U, S^n) = 
    \bigl(1-\prod_{i\in\Jc}\delta_i\bigr) \sum_{i=1}^n  H(X_i \cond Y^{i-1}_{\mathcal{J}}, U, S^{i-1})}
  \end{align}%
  from which and \eqref{eq:tmp821}, we obtain \eqref{eq:essential-erasure}.

\subsection{Proof of Lemma~\ref{lemma:order}}\label{appendix:orderRate}
We first provide the converse proof. Similarly to section \ref{section:converse}, we build on genie-aided bounds and the channel symmetry inequality. Lemma \ref{lemma:erasure-ineq}. 
Let us assume that the transmitter wishes to convey the message $W_{\Kc}$ to a subset of users $\Kc \subseteq \{1,\dots, K\}$ and receiver $j$ wishes to decode all messages $\tilde{W}_j\eqdef \{W_{\Kc}\}_{j\in \Kc}$ for $j=1, \dots, K$. The messages are all independent. 
We let $R_{\Kc}$ denote the rate of the message $W_{\Kc}$. 
In order to characterize the upper bound on the $|\Kc|$-th order message rate $R_{\Kc}$, we use genie-aided bounds by assuming that receiver $k$ provides $Y^{k}$ to receivers $k+1$ to $K$. Under this setting and using the Fano's inequality, we have for receiver 1 :
\begin{align}
n \left(\sum_{1 \in \Jc \subseteq [K] }R_{\Jc}-\epsilon \right)
&= H(Y_1^n|S^n) - H(Y_1^n |\tilde{W}_1S^n) 
\end{align}
For receiver $k=2, \dots, K$, we have: 
\begin{align}
& n \left( \sum_{k \in \Jc \subseteq \{k,\dots, K\}} R_{\Jc}-\epsilon \right) = H( \tilde{W}_k |\tilde{W}^{k-1} S^n) \\
&\leq I(\tilde{W}_k; Y_1^n \dots Y_k^n|\tilde{W}^{k-1} S^n)  \\
&= H(Y_1^n \dots Y_k^n|\tilde{W}^{k-1} S^n)- H(Y_1^n \dots Y_k^n|\tilde{W}^{k} S^n)
\end{align}
Summing up the above inequalities with appropriate weights and applying Lemma \ref{lemma:erasure-ineq} $K-1$ times, we readily obtain for this user ordering:
\begin{align}\label{eq:upperbound}
n \left( \sum_{k=1}^K \frac{\sum_{k \in \Jc \subseteq \{k,\dots, K\}} R_{\Jc}}{1-\delta^k}-\epsilon \right) &\leq \frac{H(Y_1^n|S^n)}{1-\delta} \\
& \leq 1.
\end{align}
We further impose the symmetrical rate condition such that $R_{\Kc} = R_{\Kc'}$ for any subset $\Kc, \Kc'$ with equal cardinality and define the $j$-th order message rate as $R^j(K) = R_{\Kc}$ for any $\Kc$ of cardinality $j$. By focusing on $\Jc$ of the same cardinality $j$ in \eqref{eq:upperbound}, 
the upper bound on $R^j(K)$ is given by  
\begin{align}
R^j (K) \leq \frac{1}{\sum_{k=1}^{K} \frac{ {K-k \choose j-1}}{1-\delta^k}}. 
\end{align}

In order to prove the achievability of the $i$-th order rate, we proceed GGT algorithm from phase $i$ by sending $N_i$ packets to each subset 
$\Ic \subseteq [K]$ with $|\Ic|=i$. The length of $j$-order subphase in \eqref{eq:RecursiveLength} 
is now given by 
\begin{align}\label{eq:t^i_j}
t^i_j (N_i)= \frac{1}{\beta_j} \sum_{l=i}^{j-1} {j-1 \choose l-1}N^i_{l\rightarrow j}, \;\; j> i
\end{align}
where we added the index $i$ and the dependency on $N_i$ to clarify the fact that the algorithm starts by sending $N_i$ packets in each subphase in phase $i$ with $N^i_{l\rightarrow j} = t_l^i \alpha_{l\rightarrow j}$. The dependency on $N_i$ might be omitted if it clear. 
For $j=i$, we have 
\begin{align}\label{eq:t^i_i}
t^i_i (N_i) = \frac{N_i}{\beta_i}
\end{align}
The sum rate of order-$i$ messages achieved by GGT algorithm is given by  
\begin{align}\label{eq:Rggl}
R^{i}_{\rm GGT}(K)=\frac{{K \choose i}N_i}{\sum_{j=i}^{K}{K \choose j}t^{i}_j (N_i)}~~~~\forall j.
\end{align} 
We notice that the number of transmission from phase $j$ to $K$ can be expressed by grouping subphases in the following different way:  
\begin{align}
\sum_{j=i}^{K}{K \choose j}t^{i}_j (N_i)= \sum_{j=i}^{K}U_{j}^{i}
\end{align} 
where 
\begin{align}
U_{j}^{i}=\sum_{l=i}^{j}{j-1 \choose l-1}t^{i}_l~~~~\forall j\geq i
\end{align}
By following similar steps as \cite[Appendix C]{gatzianas2013multiuser}, we obtain the recursive equation given by 
\begin{align}\label{eq:U^i_j}
U_{j}^{i}&=\frac{1}{\beta_j}\sum_{l=1}^{j-i}{j-1 \choose l}(-1)^{l+1}\beta_{j-l}U_{j-l}^{i} 
\end{align}
for $j >i$.
Since we have $U_i^i=t^i_i=\frac{N_i}{\beta_i}$ and using the equality ${j-1 \choose c}{j-c-1 \choose i-1}={j-1 \choose j-i}{j-i \choose c}$ and the binomial theorem $\sum_{k=0}^n {n \choose k} x^k y^{n-k}=(x+y)^n$, it readily follows that we have 
\begin{align}
U_j^i&=\frac{N_i}{\beta_j}{j-1 \choose j-i}, \;\;\; j\geq i.
\end{align}
By plugging the last expression in \eqref{eq:Rggl}, we have  
\begin{align}
R^{i}_{\rm GGT} (K)&=\frac{{K \choose i}N_i}{\sum_{j=i}^{K}\frac{N_i}{\beta_j}{j-1 \choose j-i}}\\
&=\frac{{K \choose i}}{\sum_{k=1}^{K-i+1}\frac{{K-k \choose i-1}}{\beta^{K-k+1}}}
\end{align} 
which coincides the upper bound of \eqref{eq:Rorder}. This establishes the achievability proof. 


\subsection{Proof of Corollary~\ref{cor:Sheng}} \label{appendix:Sheng}
We prove that the RHS of \eqref{eq:Sheng}, denoted here by $f$, coincides with $R^1(K)$ by exploiting the result of Lemma \ref{lemma:order}. 
By replacing $R^i(K)$ with $R^i_{\rm GGT}(K)$ defined in \eqref{eq:Rggl} for $N_i=N_{1\rightarrow i}$ for $i\geq 2$ and $N_1=N_{1\rightarrow 1}$, we have
\begin{align}
f & =\frac{KN_1}{\frac{KN_1}{\beta_1}+\sum_{i=2}^{K}\sum_{j=i}^{K}{K \choose j}t_j^{i} (N_{1\rightarrow i})}\\
&=\frac{KN_1}{\frac{KN_1}{\beta_1}+\sum_{j=2}^{K}{K \choose j}\sum_{i=2}^{j}t_j^{i} (N_{1\rightarrow i})}
\end{align}
To prove that $f= R^{1}_{\rm GGT}(K) = \frac{KN_{1\rightarrow j}}{\sum_{j=1}^K {K \choose j} t_j}$ where $t_j$ is defined in \eqref{eq:RecursiveLength}, it suffices to prove the following equality:
\begin{align}\label{eq:t_j}
t_j(N_1) &=\sum_{i=2}^j t^i_j (N_{1\rightarrow i}) ~~~~~\forall j \geq 2
\end{align}
For $j=2$, the above equality follows from \eqref{eq:RecursiveLength} and \eqref{eq:t^i_i}.
\begin{align}
t_2(N_1) =\frac{N_{1\rightarrow 2}}{\beta_2}=t_2^{2}(N_{1\rightarrow 2}).
\end{align}
Now suppose that \eqref{eq:t_j} holds true for $2 \leq l \leq  j-1$ and we prove it for $j$. From \eqref{eq:RecursiveLength} we have
\begin{align}
t_j(N_1) &=\frac{1}{\beta_j}\sum_{l=1}^{j-1}{j-1 \choose l-1}N_{l\rightarrow j} \\
&=\frac{1}{\beta_j}\left[\sum_{i=2}^{j-1}\sum_{l=i}^{j-1}{j-1 \choose l-1}N^i_{l\rightarrow j}+N_{1\rightarrow j} \right]\\
&=\sum_{i=2}^{j-1}t^{i}_j+ t_j^j 
\end{align}
where the second equality follows by recalling $N_{l\rightarrow j}=t_l \alpha_{l\rightarrow j}$ and plugging the recursive expression \eqref{eq:t_j} in $t_l$ for 
$l=2,\dots, j-1$, the last equality is due to \eqref{eq:t^i_i}. Therefore, 
we verify the desired equality also for $j$. 
This yields
\begin{align}
f &=\frac{KN_1}{Kt_1^{1}(N_1)+\sum_{j=2}^{K}{K \choose j}\sum_{i=2}^{j}t_j^{i}(N_{1\rightarrow j})}\\
&=\frac{KN_1}{Kt_1^{1}(N_1)+\sum_{j=2}^{K}{K \choose j}t_j^{1}(N_{1\rightarrow j})}\\
&=\frac{KN_1}{\sum_{j=1}^{K}{K \choose j}t_j^{1}(N_1)}\\
&=R^1_{\rm GGT}(K).
\end{align}


\end{document}